\documentclass[11pt]{article}
\usepackage{amsmath}
\usepackage{float}
\usepackage{url}
\usepackage{xcolor}
\usepackage{enumerate}
\usepackage{colortbl}   
\usepackage[utf8]{inputenc} 
\usepackage[T1]{fontenc}    
\usepackage{url}            
\usepackage{booktabs}       
\usepackage{amsfonts}       
\usepackage{nicefrac}       
\usepackage{microtype}      
\usepackage{lipsum}		
\usepackage{graphicx}
\usepackage{natbib}
\usepackage{doi}
\newtheorem{theorem}{Theorem}

\newtheorem{proof}{Proof}
\newtheorem{assumption}{Assumption}
\newtheorem{example}{Example}

\title{Subsampling-based Tests in Mediation Analysis}
\author{ Asmita Roy$^1$, Huijuan Zhou $^2$, Ni Zhao $^1$, Xianyang Zhang $^3$\footnote{Corresponding author: zhangxiany@stat.tamu.edu} \smallskip \\
	$^1$Department of Biostatistics, Johns Hopkins University\\
	$^2$School of Statistics and Management, \\Shanghai University of Finance and Economics\\
	$^3$ Department of Statistics, Texas A\&M University}
    \date{}
\begin{document}
{\centering\maketitle}
\begin{abstract}
Testing for mediation effect poses a challenge since the null hypothesis (i.e., the absence of mediation effects) is composite, making most existing mediation tests quite conservative and often underpowered. In this work, we propose a subsampling-based procedure to construct a test statistic whose asymptotic null distribution is pivotal and remains the same regardless of the three null cases encountered in mediation analysis. The method, when combined with the popular Sobel test, leads to an accurate size control under the null. We further introduce a Cauchy combination test to construct p-values from different subsample splits, which reduces variability in the testing results and increases detection power. Through numerical studies, our approach has demonstrated a more accurate size and higher detection power than the competing classical and contemporary methods.
\end{abstract}


\section{Introduction}
Mediation analysis serves as a powerful statistical framework for elucidating causal relationships along biological pathways in the human body. A compelling example lies in the mechanism of Metformin, a widely prescribed anti-diabetic medication, which influences inflammation markers in the blood through its effects on serum Short Chain Fatty Acids (SCFAs) \citep{vinolo2011regulation,mueller2021metformin}.
Let $\alpha$ denote the exposure's causal effect on the mediator, and $\beta$ represent the mediator's causal effect on the outcome, conditional on the exposure. Under classical linear structural equation modeling and appropriate identification assumptions, the causal mediation effect is proportional to the product $\alpha\beta$ \citep{imai2010general}. This product-based representation extends beyond linear models to various other contexts, including generalized linear models and survival models \citep{huang2016mediation,vanderweele2010odds,vanderweele2011causal}.
A fundamental question in mediation analysis is to examine the existence of a mediation effect, which can be formulated as testing the null hypothesis $H_0:\alpha\beta=0$ against the alternative hypothesis $H_a:\alpha\beta\neq 0$. This null hypothesis encompasses three distinct scenarios: $H_0=H_{01}\cup H_{10}\cup H_{00}$, where (i) $H_{01}: \alpha=0,\beta\neq0$, (ii) $H_{10}: \alpha\neq 0,\beta=0$ and (iii) $H_{00}: \alpha=\beta=0$.
Traditional mediation tests, such as the Sobel test \citep{sobel1982asymptotic} and MaxP test \citep{mackinnon2002comparison}, rely on asymptotic distributions derived from scenarios (i) and (ii). However, these approaches can yield overly conservative results when analyzing genomics data, as most omics markers fall under scenario (iii). To address this limitation, we propose a modified version of the classical Sobel test. The key insights and steps behind the construction of our testing procedure can be summarized as follows:
\begin{itemize}
    \item Develop a test statistic for the mediation effect whose asymptotic null distribution belongs to the same parametric family under all null scenarios ($H_{00}$, $H_{01}$, and $H_{10}$), differing only in their scaling parameters.
    \item Implement a subsampling-based approach to construct a scale-invariant studentized statistic whose asymptotic null distribution remains the same across all null scenarios. This ensures the statistic's distribution is independent of whether the null hypothesis falls under $H_{00}$, $H_{01}$ or $H_{10}$.
    \item Apply the Cauchy combination method to aggregate p-values obtained from various subsample splits, yielding a stable and powerful test.
\end{itemize}
Through comprehensive numerical studies, we demonstrate that our proposed test, named the Cauchy-combined Studentized Mediation Test or CSMT, achieves both superior size control and enhanced power compared to both classical and contemporary approaches. The complete R implementation for reproducing our simulation results is freely available at \url{https://github.com/asmita112358/CSMT}.

\section{Sobel's test}
Consider a general structural equation model where coefficient $\alpha$ represents the causal effect of the exposure on the mediator, while coefficient $\beta$ quantifies the mediator's effect on the outcome, conditional on the exposure. We impose the following assumptions, which hold for most structural equation models.
\begin{assumption}\label{as1}
    The coefficients $\alpha$ and $\beta$ admit consistently asymptotically normal estimators $\hat{\alpha}$ and $\hat{\beta}$, with the following asymptotic distribution
    $$n^{1/2}\begin{pmatrix}
    \hat{\alpha} - \alpha\\
    \hat{\beta} - \beta
\end{pmatrix} \implies  N\left( \begin{pmatrix}
    0\\
    0
\end{pmatrix},
\begin{pmatrix}
    \sigma_1^2 &0\\
    0 &\sigma_2^2
\end{pmatrix}\right),$$
where $n$ denotes the sample size and ``$\implies$" denotes weak convergence.
\end{assumption}
\begin{assumption}\label{as2}
The variances $\sigma^2_1$ and $\sigma^2_2$ admit consistent estimators $s^2_1$ and $s^2_2$.
\end{assumption}

Let $T_{n,
\alpha} = n^{1/2}\hat{\alpha}/s_1$ and $T_{n,
\beta} = n^{1/2}\hat{\beta}/s_2$ be the t-statistics for testing $\alpha=0$ and $\beta=0$ respectively. 
The Sobel statistic is defined as 
$$S_n = \frac{n^{1/2}\hat{\alpha}\hat{\beta}}{(\hat{\alpha}^2s^2_2  + \hat{\beta}^2s^2_1)^{1/2}},$$
which can be re-written in terms of the t-statistics, $T_{n,\alpha}$ and $T_{n,\beta}$, as 
\begin{align}\label{sn}
S_n =  \frac{T_{n,\alpha}T_{n,\beta}}{(T_{n,\alpha}^2 + T_{n,\beta}^2)^{1/2}}. 
\end{align}
We present the following result concerning the asymptotic distribution of the Sobel test under the composite null, which has been proved in \cite{glonek1993behaviour}. For completeness, we present the proof below.

\begin{theorem}\label{thm1}
Under Assumptions \ref{as1}-\ref{as2}, $S_n$ converges weakly to a normal distribution with mean 0 and variance $\tau$, where
    \begin{align*}
        \tau = 
        \begin{cases}
        1,  & \alpha=0,\beta\neq 0 \text{ or } \alpha\neq 0,\beta=0,\\
        1/4, & \alpha=\beta=0.
        \end{cases}
    \end{align*}
\end{theorem}

\begin{proof}
Under Assumptions \ref{as1}-\ref{as2} and $H_{01}$ or $H_{10}$, by the multivariate delta method,
$$ n^{1/2} \hat{\alpha}\hat{\beta} = n^{1/2}( \hat{\alpha}\hat{\beta} - \alpha\beta) = n^{1/2}\hat{\alpha}(\hat{\beta}-\beta)+ n^{1/2}\beta(\hat{\alpha}-\alpha)\implies N(0, \alpha^2\sigma^2_{2} + \beta^2\sigma^2_{1}),$$
where the variance can be consistently estimated by $\hat{\alpha}^2 s_1^2 + \hat{\beta}^2s_2^2$. Hence by slutsky's theorem, $S_n$ converges weakly to $N(0,1)$ under $H_{01}$ or $H_{10}$.

When $\alpha = \beta = 0$, the multivariate delta method cannot be used. In this case, we note that $T_{n,\alpha}$ and  $T_{n,\beta}$ converges weakly to $N(0,1)$ and they are asymptotically independent. By (\ref{sn}), the asymptotic distribution of the Sobel statistic is equivalent to the distribution of 
$$W = \frac{Z_1Z_2}{(Z_1^2 + Z_2^2)^{1/2}},$$ 
where $Z_1$ and $Z_2$ are independent standard normal random variables. Consider the polar transformation: $Z_1 = R\cos(\theta)$ and $Z_2 = R\sin(\theta)$, where $R$ follows a Rayleigh distribution, $\theta$ follows the uniform distribution over $(0,2\pi)$ (denoted by $U(0,2\pi)$),
and $R$ and $\theta$ are independent with each other. Thus, $W$ can be written as
$W = R\sin(2\theta)/2.$
Since $\theta \sim U(0, 2\pi)$, $\sin(2\theta)$ has the same distribution as that of $\sin(\theta)$. Therefore, $R\sin (2\theta) \sim N(0,1)$ and $W \sim N(0, 1/4)$.
\end{proof}


\section{A Subsampling-based test}\label{sec-t-test}
We introduce a new methodology to address the challenge of obtaining a test with accurate size under the composite null of no mediation effect. Our idea is to partition the data into $K$ subsamples, namely $\mathcal{G}_1,\ldots,\mathcal{G}_K$ such that
$\mathcal{G}_i\cap \mathcal{G}_j=\emptyset$ for $i\neq j$ and $\cup_{i=1}^K\mathcal{G}_i=[n]:=\{1,\ldots,n\}$. For each subsample $\mathcal{G}_i$, we compute the corresponding Sobel statistic denoted by $S_{\mathcal{G}_i}$ for $i=1,\ldots,K$. Let $\bar{S}_K=\sum^{K}_{i=1}S_{\mathcal{G}_i}/K$.
Depending on the true state of the composite null hypothesis, we have
\begin{align*}
S_{\mathcal{G}_i}\implies\begin{cases}
N(0,1) & \text{ if } \alpha=0,\beta\neq 0 \text{ or } \alpha\neq 0,\beta=0,\\
N(0,1/4) & \text{ if } \alpha=\beta=0,
\end{cases}
\end{align*}
and $S_{\mathcal{G}_i}$s are asymptotically independent. We define the subsampling-based t-statistic as 
\begin{align*}
\mathcal{T}_n = \frac{K^{1/2}\bar{S}_K}{\left\{\frac{1}{K-1}\sum^{K}_{i=1}(S_{\mathcal{G}_i}-\bar{S}_K)^2\right\}^{1/2}}.  
\end{align*}
\begin{theorem}\label{thm2}
Under $H_{0}$ irrespective of the null types and for fixed $K$, we have 
    \begin{align*}
        \mathcal{T}_n \implies\frac{K^{1/2}\bar{Z}_K}{\left\{\frac{1}{K-1}\sum^{K}_{i=1}(Z_i-\bar{Z}_K)^2\right\}^{1/2}}\sim t_{K-1}.
    \end{align*}
    where $Z_1,\dots,Z_K \overset{i.i.d}{\sim} N(0,1)$, $\bar{Z}_K=\sum^{K}_{i=1}Z_i/K$ and $t_{K-1}$ represents the $t$ distribution with $K-1$ degrees of freedom. 
\end{theorem}
\begin{proof}
The result follows by applying Theorem \ref{thm1} to each subsample and the continuous mapping theorem. 
\end{proof}
In view of theorem \ref{thm2},
we reject $H_0$ if $|\mathcal{T}_n| > q_{K-1}(1-\alpha/2)$, where $q_{K-1}(1-\alpha/2)$ is the $1-\alpha/2$ quantile of $t_{K-1}$. The corresponding p-value can be defined as $p=P(|t_{K-1}|>|\mathcal{T}_n|)$. In practice, subsamples are created by randomly permuting the data and generating \( K \) splits. The sizes of each subsample are maintained as closely equal as possible. If the total sample size \( n \) is not divisible by \( K \), the remaining samples (of the size $n- K\lfloor n/K\rfloor$) are randomly redistributed among the subsamples.  For example, when $n=200$ and $K=7$, we have 4 subsamples of size 29 and 3 subsamples of size 28.


\section{Cauchy Combination for Repeated Sample Splitting}
The subsampling-based t-test is asymptotically accurate under the three null types; however, it has two notable issues: (i) the test results are influenced by how the data is split, and (ii) our simulations have empirically shown that the subsampling-based t-test may experience a loss of power in certain situations. To address these challenges, we shall employ a repeated subsampling procedure and combine the p-values from different splits using the Cauchy combination test \citep{liu2020cauchy}. Specifically, suppose we partition the data $M$ times and calculate the p-values $\{p_m \}_{m = 1}^M$ for each of those $M$ partitions. Then, the Cauchy combination statistic based on these $M$ p-values is given by 
\begin{align*}
C_m=\sum_{m = 1}^M w_m \tan (\pi (0.5 - p_m)),
\end{align*}
where $w_m$s are non-negative weights satisfying that $\sum_{m = 1}^M w_m= 1.$ In practice, we suggest setting $w_m=u_m/\sum^{M}_{i=1}u_i$, where $u_i$s are independently generated from $U(0,1)$.
We reject $H_0$ when $C_m>c_{1-\alpha}$ and the corresponding Cauchy combined p-value is given by $\text{pr}(C > C_m)=0.5-\arctan(C_m)/\pi$, where $c_{1-\alpha}$ is the $1-\alpha$ quantile of the standard Cauchy distribution, and $C$ follows a standard Cauchy distribution. As will be shown in Section \ref{sec:sim}, this testing procedure leads to improved power while keeping the size under control as compared to existing approaches.

\section{Choice of K}\label{seck}
The selection of \( K \) is critical to our procedure. Choosing a smaller \( K \) may lead to greater variation in the t-statistic, while a larger \( K \) can reduce the sizes of the subsamples, potentially harming the performance of the Sobel test. We recommend setting \( K = \lfloor 0.5 n^{1/2} \rfloor \) based on empirical evidence. This choice moderates the size across various settings while maintaining sufficiently high power. Simulations supporting our recommendation are presented in the Supplementary Materials; see Figure \ref{var_K2} therein for details.

\section{Numerical studies}\label{sec:sim}
We compare CSMT with three existing approaches, including the classical Sobel's test 
\citep{sobel1982asymptotic}, the MaxP test \citep{mackinnon2002comparison} and the adaptive bootstrap based test in \citep{10.1093/jrsssb/qkad129} (denoted by ABtest). We generate data from the following structural equation model:
\begin{equation}\label{simmodel}
    \begin{aligned}
    &G = \alpha S + \alpha_0 + \alpha_1 X_1 + \alpha_2 X_2 + \epsilon,\\
    &Y = \beta G + \beta_0 + \beta_1 X_1 + \beta_2 X_2 + \tau S + e,
\end{aligned}
\end{equation}
where $(X_1,X_2)$ are independently generated from $U(0,1)$, and $e$ and $\epsilon$ are independently generated from $N(0,1)$.
We report the numerical results based on 500 independent simulation runs. The proportions of the four types of hypotheses $H_{ij}$ in the simulation runs are denoted by $\pi_{ij}$ for $0 \leq i,j \leq 1$. Throughout, we set 
$M = 500$ in the Cauchy combination test.

\begin{example}
To study the size, we consider two sets of proportions: the dense null and the sparse null. In the dense null setup, we have $\pi_{00} = 0.4$, $\pi_{10} = 0.3$, $\pi_{01} = 0.3$, and $\pi_{11} = 0$. In the sparse null setup, the proportions are $\pi_{00} = 0.8$, $\pi_{10} = 0.1$, $\pi_{01} = 0.1$, and $\pi_{11} = 0$. 
For each simulation run, we generate $n = 600$ samples from Model \eqref{simmodel} under the hypotheses $H_{00}, H_{01}$, and $H_{10}$. We consider $(\alpha, \beta) = (0,0)$ under $H_{00}$, $(\alpha, \beta) = (r,0)$ under $H_{10}$ and $(\alpha, \beta) = (0,r)$ under $H_{01}$, where $r=0.1$ or 0.5. The value of $K$ has been consistently set to $\lfloor 0.5n^{1/2}\rfloor=12$.
The QQ plots of the negative of the common logarithm of the p-values and the bar plots of the empirical sizes for \( r = 0.1 \) are presented in Figures \ref{c1-q} and \ref{c1-b}, respectively. A dotted line segment marks the position of $-\log_{10}(0.05)$ along the $x$ and $y$ axes in the QQ plot. It is observed that the ABtest exhibits an inflated size, while the MaxP and Sobel tests appear to be conservative, showing non-uniform p-values under the two types of nulls. 
The p-value distribution of CSMT is closer to the uniform distribution, especially for the low quantile levels, which results in a more accurate size than the other methods. The corresponding figures for \( r = 0.5 \) have been presented in the Supplementary materials, (Figures \ref{c2-q} and \ref{c2-b}) where a similar pattern can be seen across the four different methods.

\begin{figure}[H]
    \centering
    \includegraphics[width = 0.8\linewidth]{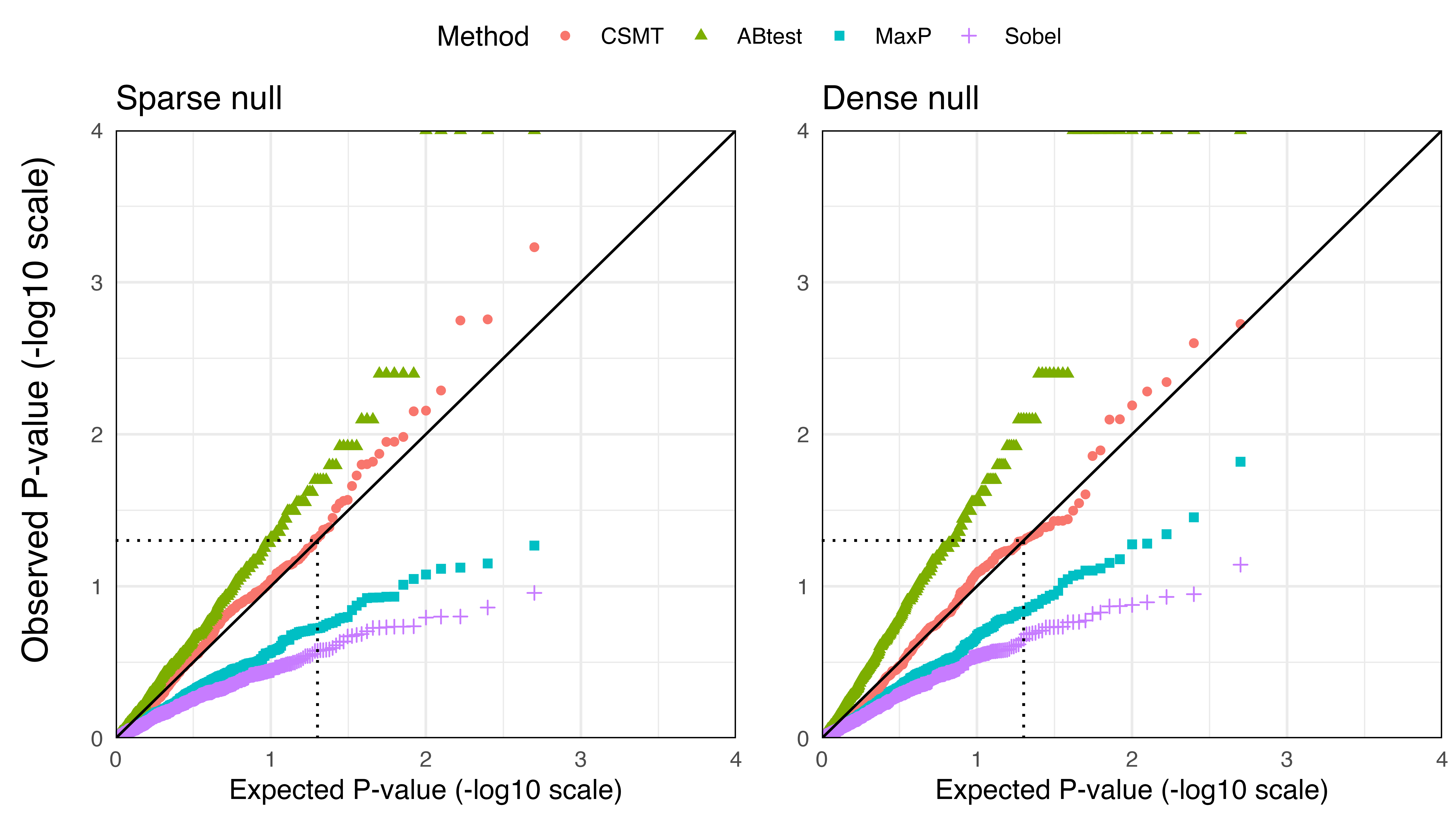}
    \caption{QQ plots of the p-values in log scale from the four competing methods under the sparse null (left) and dense null (right) with $r = 0.1$. The nominal level is 5\%.
    }
    \label{c1-q}
\end{figure}

\begin{figure}[H]
    \centering
    \includegraphics[width = 0.8\linewidth]{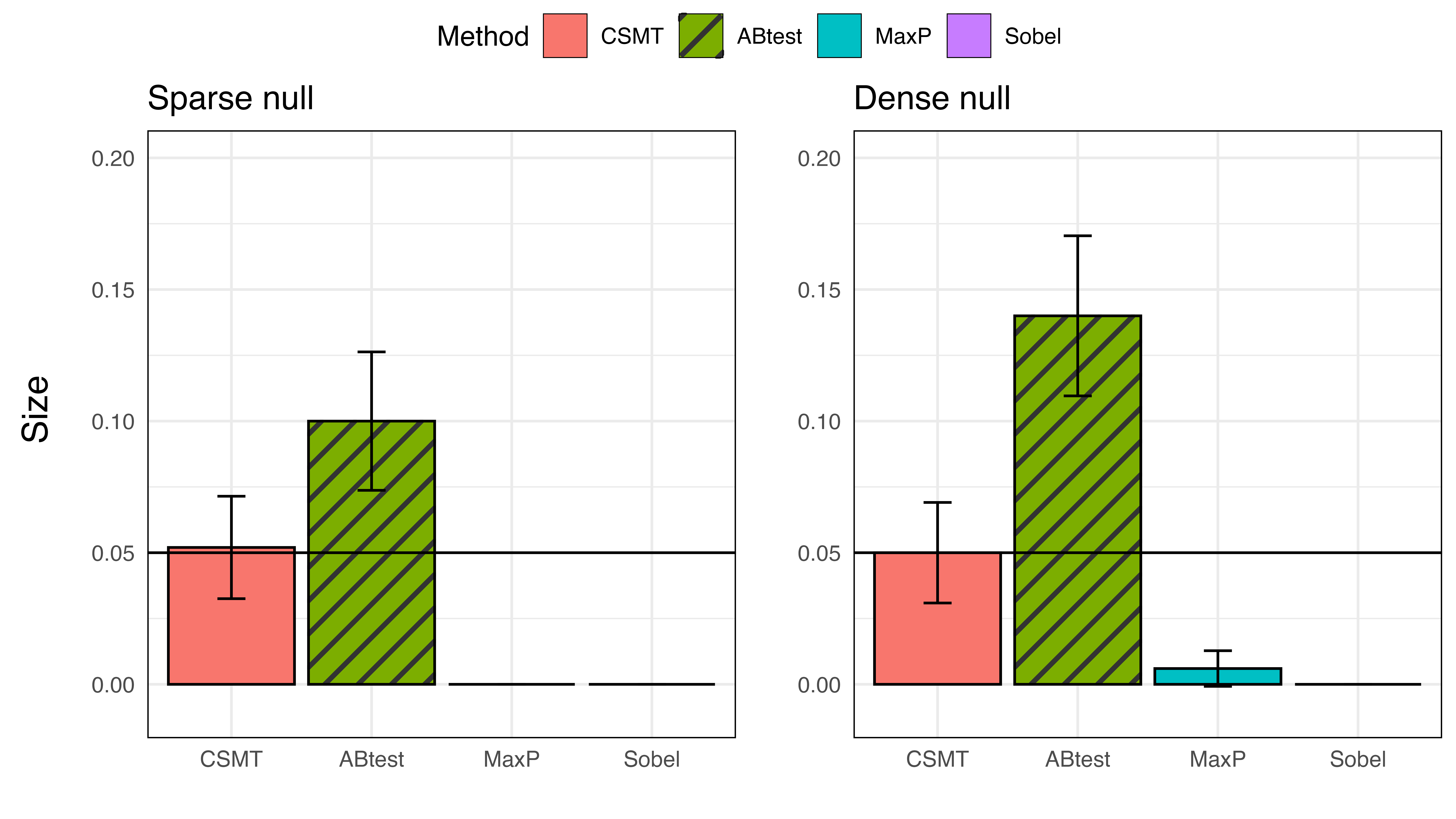}
    \caption{Empirical sizes for the four competing methods under the sparse null (left) and dense null (right) with $r = 0.1$. The nominal level is 5\%, represented by a horizontal line.}
    \label{c1-b}
\end{figure}

\end{example}

\begin{example}
For power comparison, we employ a setup similar to that described in \cite{10.1093/jrsssb/qkad129}. The data is simulated under two scenarios: (i) We fix \(\alpha = \beta\) and vary this value from 0.1 to 0.5. For each fixed value of \(\alpha\) and \(\beta\), we generate samples of sizes 100, 200, and 300, reporting the empirical power based on 500 independent simulation runs.  (ii) We set \(\alpha \beta = 0.3\) for \(n = 100\), \(\alpha \beta = 0.2\) for \(n = 200\), and \(\alpha \beta = 0.1\) for \(n = 300\). The empirical rejection rates, calculated from 500 independent simulations, are plotted against the ratio \(\alpha/\beta\). Plots for both scenarios across all sample sizes are provided in Figure \ref{fig:allpower}. The value of \(K\) has been determined based on the recommendation in Section \ref{seck}. The Sobel and MaxP tests are outperformed by the ABtest and CSMT. Interestingly, the power of the ABtest decreases as the signal strength increases, but it exhibits the highest power compared to all other methods at lower signal strengths. Here we would like to set a reminder that this improved power comes at the cost of inflated size. In contrast, CSMT shows the best power at higher signal strengths and performs well even with lower sample sizes.

\begin{figure}[H]
    \centering
    \includegraphics[width = \linewidth]{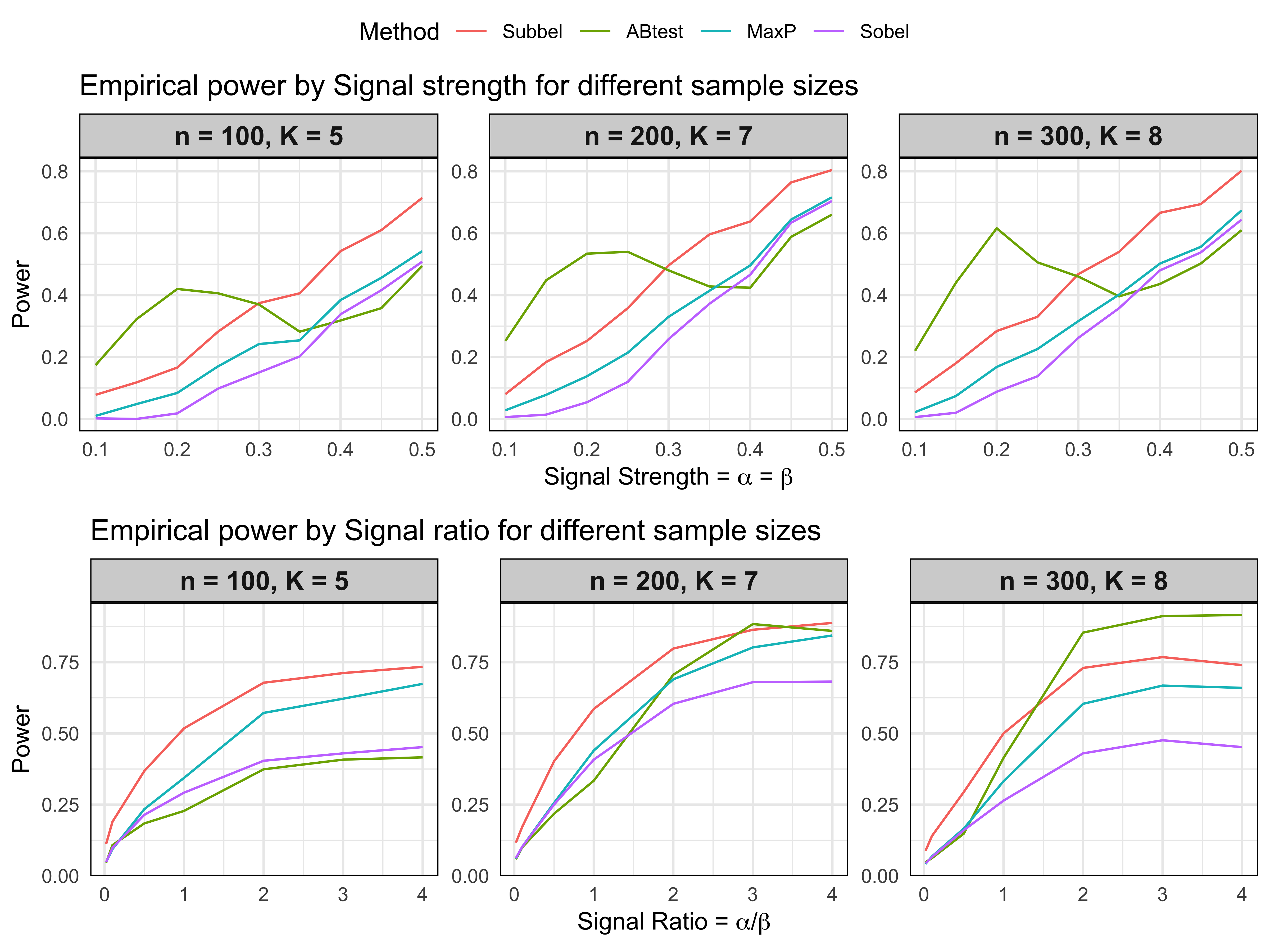}  
    \caption{Empirical powers for the four competing methods with $n\in\{100,200,300\}$ and varying signal strength (top) and varying signal ratio (bottom). The nominal level is 5\%.}
    \label{fig:allpower}
\end{figure}
\end{example}

\section{Real Data Analysis}
We analyzed a clinical data set from the SPIRIT trial (Survivorship Promotion in Reducing IGF‐1) \cite{Yeh2020, Yeh2021}, a randomized three-parallel-arm controlled clinical trial. In this trial, cancer survivors with overweight or obesity were randomized to three groups: (1) self-directed weight loss (control), (2) coach-directed weight loss, and (3) metformin treatment. For this study, we compared the metformin-treated group (42 participants) with the control group (40 participants), based on literature supporting metformin's impact on the gut microbiome, which plays a crucial role in fermenting indigestible dietary fibers into short-chain fatty acids (SCFAs) \citep{mueller2021metformin,de2017metformin}. SCFAs are vital substances known to regulate inflammation \citep{vinolo2011regulation,tedelind2007anti}.

In this analysis, we investigated the mediating role of SCFAs, measured at six months into the trial, on serum inflammation markers assessed at twelve months. Our analysis focused on two key inflammation markers, Interleukin-6 (IL-6) and C-reactive protein (CRP), and seven serum SCFAs. Elevated levels of IL-6 are indicative of inflammatory conditions, while increased CRP levels are associated with a heightened risk of heart disease and stroke. We tested the treatment-mediator-outcome hypothesis for all possible combinations of SCFAs and inflammation markers, with p-values reported in Table \ref{tab:SPIRIT}. Notably, Butyric Acid, Acetic Acid, and Valeric Acid were identified as significant mediators in regulating inflammation markers when treated with metformin. Interestingly, increased serum and fecal levels of butyrate have been found to lower blood pressure and reduce hypertension \citep{tilves2022increases}. Additionally, supplementation with sodium acetate and butyrate has been shown to decrease inflammatory markers in mice \citep{chen2022sodium}, which supports our findings. Overall, CSMT identifies more SCFA-inflammation marker pathways than competing methods.

\begin{table}[H]
\caption{\label{tab:SPIRIT} P-values for different mediator/outcome combinations for the SPIRIT data set, where we highlight the numbers in red if $p \leq 0.05$, in orange if $0.05 < p \leq 0.1$ and in yellow if $0.1 < p \leq 0.2$.}
\centering
\begin{tabular}[t]{|l|l|r|r|r|r|}
\hline
Outcome & Mediator & CSMT & ABtest & MaxP & Sobel\\
\hline
\hline
Interleukin-6 & Acetic Acid & 0.70 & 0.43 & 0.63 & 0.64\\
\hline
Interleukin-6 & Propionic Acid & 0.45 & 0.26 & 0.53 & 0.57\\
\hline
Interleukin-6 & Isobutyric Acid & 0.30 & 0.60 & 0.69 & 0.71\\
\hline
Interleukin-6 & Butyric Acid & \cellcolor{orange}0.09 & 0.34 & 0.47 & 0.49\\
\hline
Interleukin-6 & Methylbutyric Acid & 0.76 & 0.71 & 0.88 & 0.88\\
\hline
Interleukin-6 & Valeric Acid & \cellcolor{yellow}0.11 & \cellcolor{yellow}0.15 & 0.41 & 0.48\\
\hline
Interleukin-6 & Hexanoic Acid & 0.70 & 0.63 & 0.78 & 0.78\\
\hline
C-reactive Protein & Acetic Acid & \cellcolor{orange}0.09 & 0.46 & 0.76 & 0.76\\
\hline
C-reactive Protein & Propionic Acid & \cellcolor{yellow}0.16 & 0.24 & 0.21 & 0.34\\
\hline
C-reactive Protein & Isobutyric Acid & 0.97 & 0.47 & 0.69 & 0.72\\
\hline
C-reactive Protein & Butyric Acid & \cellcolor{yellow}0.17 & 0.30 & 0.27 & 0.33\\
\hline
C-reactive Protein & Methylbutyric Acid & 0.98 & 0.74 & 0.88 & 0.88\\
\hline
C-reactive Protein & Valeric Acid & \cellcolor{red}0.04 & \cellcolor{yellow}0.20 & \cellcolor{yellow}0.17 & 0.27\\
\hline
C-reactive Protein & Hexanoic Acid & 0.67 & 0.73 & 0.81 & 0.85\\
\hline
\end{tabular}
\end{table}

\section{Acknowledgement}
The authors thank Drs. Noel Mueller and Curtis Tilves from the University of Colorado, and Dr. Jessica Yeh from Johns Hopkins University, for providing the clinical data from the SPIRIT trial.

\section{Supplementary Materials}
The supplementary materials contain some additional numerical results.


\subsection{Additional numerical results for Example 1}
We present the numerical results for Example 1 of the main paper with $r=0.5.$ The QQ plots of the p-values and the bar plots of the empirical sizes are shown in Figures \ref{c2-q} and \ref{c2-b}, respectively. It is noted that the AB test exhibits an inflated size, while the MaxP and Sobel tests are conservative, displaying non-uniform p-values under both types of null hypotheses. Additionally, the p-value distribution of CSMT is closer to a uniform distribution at low quantile levels, resulting in a more accurate size compared to the other methods.

\begin{figure}[H]
    \centering
    \includegraphics[width = 0.8\linewidth]{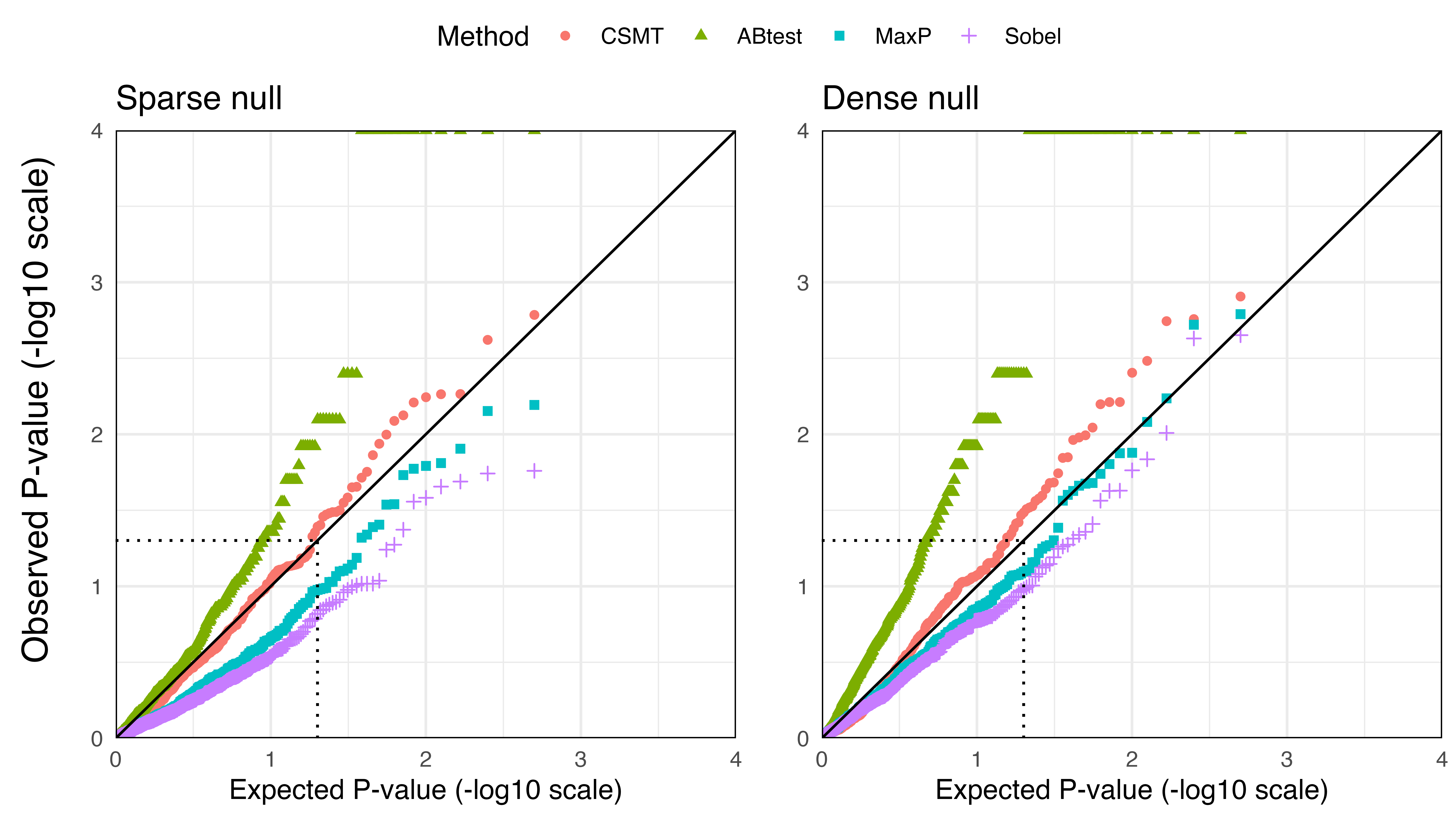}
    \caption{QQ plots of the p-values from the four competing methods under the sparse null (left) and dense null (right) with $r = 0.5$. The nominal level is 5\%.}
    \label{c2-q}
\end{figure}

\begin{figure}[H]
    \centering
    \includegraphics[width = 0.8\linewidth]{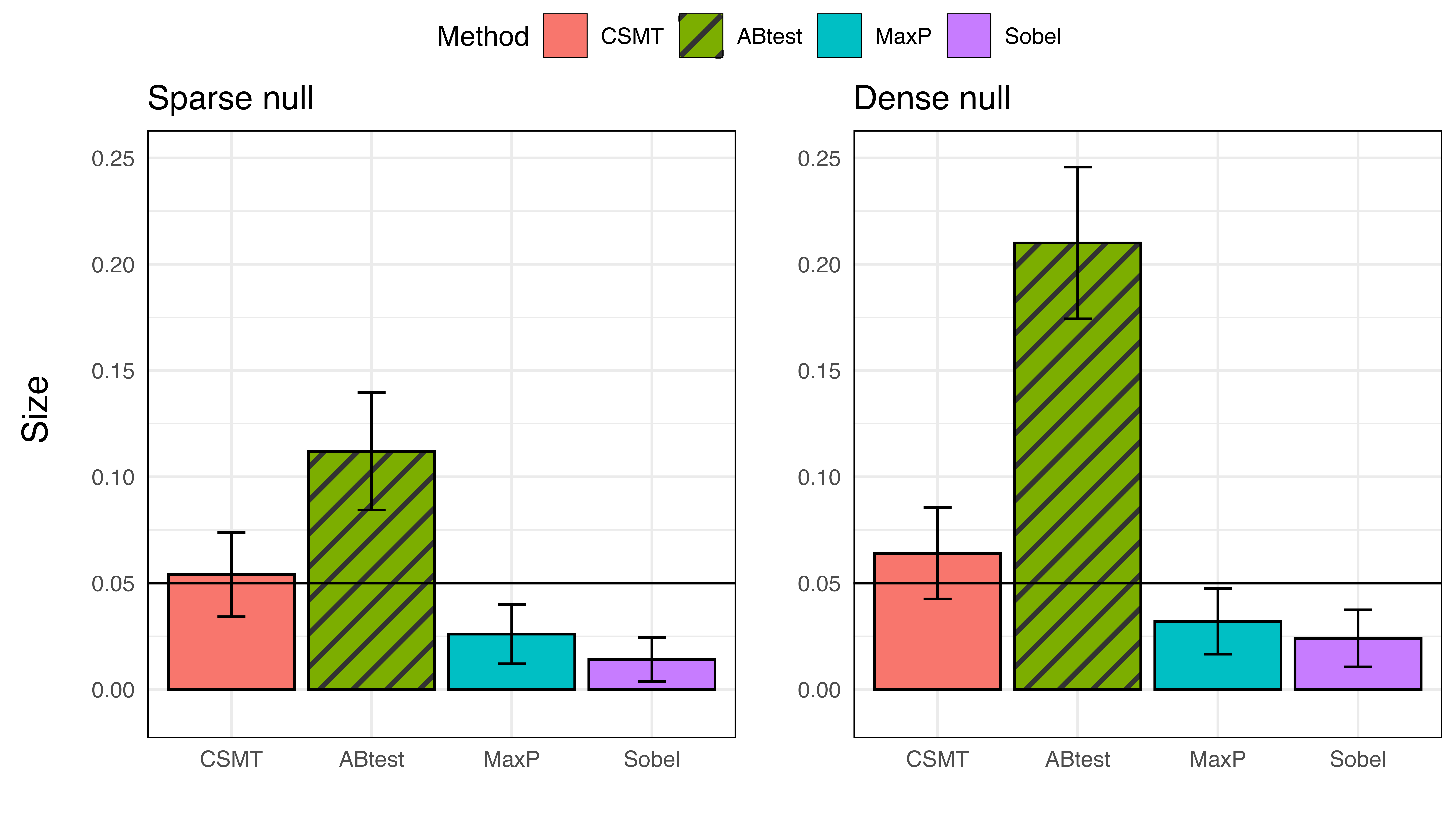}
    \caption{Empirical sizes for the four competing methods under the sparse null (left) and dense null (right) with $r = 0.5$. The nominal level is 5\%.}
    \label{c2-b}
\end{figure}

\subsection{Empirical size and power for CSMT with an adaptive choice of K}
In this section, we examine the performance of CSMT with \( K = \lfloor 0.5 n^{1/2} \rfloor \). We follow the settings outlined in Examples 1-2 of the main paper for \( n \in \{150, 250, 350, 450, 550\} \). Figure \ref{var_K2} shows that CSMT consistently demonstrates the most accurate size across all scenarios. In contrast, the ABtest exhibits an inflated size, while the MaxP and Sobel tests are undersized across all cases. Regarding empirical power, CSMT consistently outperforms both the MaxP and Sobel tests and can surpass the ABtest when the signal is strong. The ABtest remains the most powerful for weak signals.



\begin{figure}[H]
    \centering
    \includegraphics[width=\linewidth]{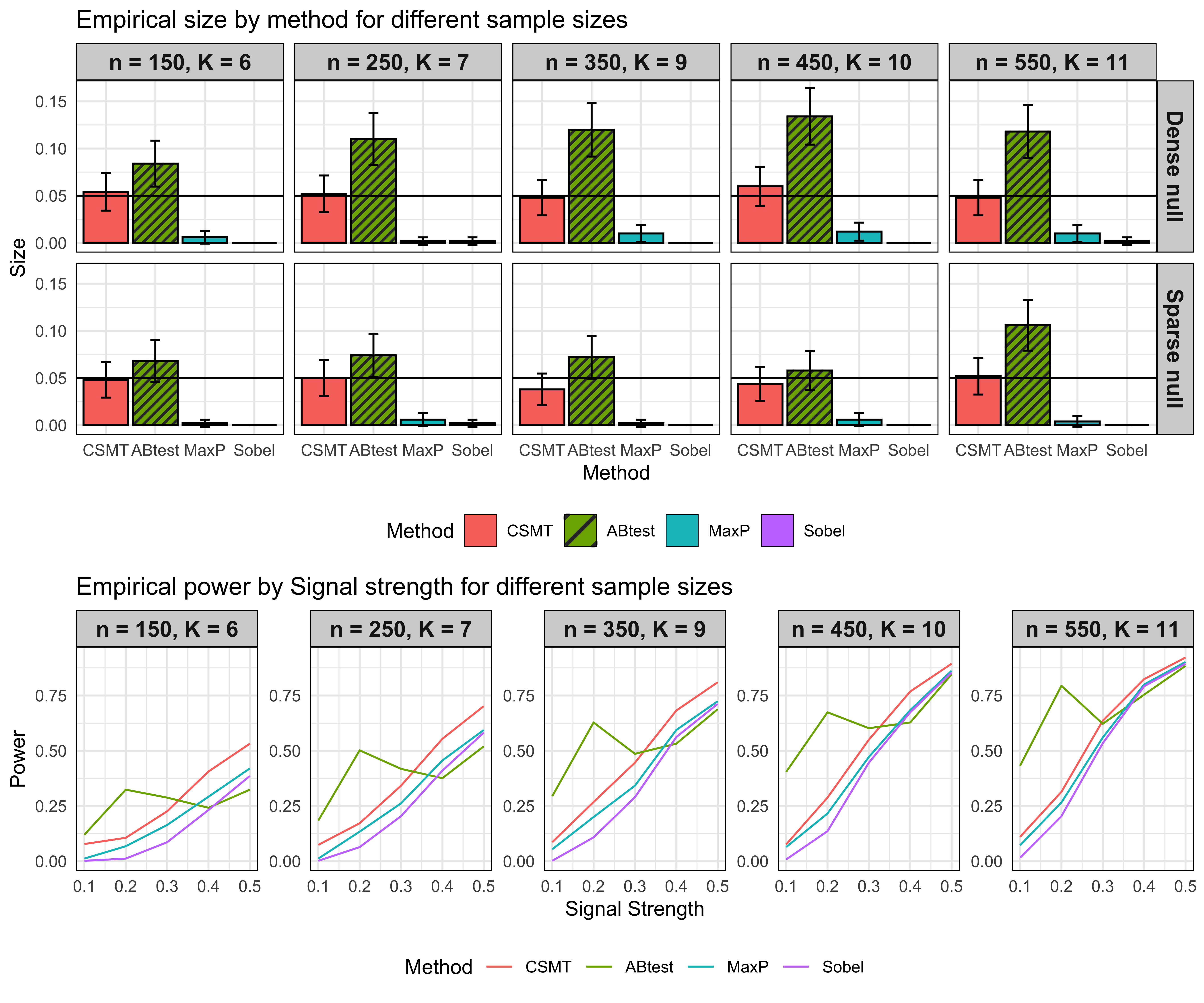}
    \caption{Empirical sizes and powers for the four competing methods with $n\in\{150,250,350,450,550\}$, where $K=\lfloor 0.5 n^{1/2} \rfloor$ for CSMT.}
    \label{var_K2}
\end{figure}

\bibliographystyle{unsrtnat}
\bibliography{reference} 
\end{document}